\documentclass[letterpaper, 10pt, conference]{ieeeconf}

\IEEEoverridecommandlockouts                              
\usepackage{comment}
\usepackage{booktabs} 

\usepackage{hyperref}
\usepackage{diagrams}
\usepackage{color}
\usepackage{cite}
\usepackage{amssymb,latexsym,amsfonts,amsmath}
\renewcommand{\theequation}{\arabic{equation}}
\usepackage{graphicx,xcolor}
\usepackage{subcaption}
\usepackage{stfloats}

\let\labelindent\relax
\usepackage{enumitem}
\setlist[enumerate,1]{label=$\bullet$,leftmargin=5.5mm}

\def\diag{\textup{diag}}

\newtheorem{theorem}{Theorem}[section]

\newtheorem{problem}[theorem]{Problem}
\newtheorem{proposition}[theorem]{Proposition}
\newtheorem{corollary}[theorem]{Corollary}

\newtheorem{definition}[theorem]{Definition}
\newtheorem{example}[theorem]{Example}

\newtheorem{remark}[theorem]{Remark}
\newtheorem{assumption}[theorem]{Assumption}

\newcommand{\n}{\mathcal{N}}

\newcommand{\C}{\mathcal{C}}

\newcommand{\N}{\mathbb{N}}%
\newcommand{\R}{\mathbb{R}}%

\title{Compositional Synthesis of Signal Temporal Logic Tasks via Assume-Guarantee Contracts}

\author{Siyuan Liu$^{1,2}$, Adnane Saoud$^3$, Pushpak Jagtap$^4$, Dimos V. Dimarogonas$^5$, Majid Zamani$^{6,2}$
\thanks{The work is supported in part by the German Research Foundation (DFG) under grant ZA 873/7-1, 
the H2020 ERC Starting Grant
AutoCPS (grant agreement No. 804639),
the ERC LEAFHOUND Project, the Swedish Research Council (VR), and the Wallenberg AI, Autonomous Systems and Software Program (WASP) funded by the Knut and Alice Wallenberg (KAW) Foundation.}
\thanks{$^1$Department of Electrical and Computer Engineering, Technical University of Munich, Germany.
$^2$Computer Science Department, LMU Munich, Germany. 
$^3$Laboratoire des Signaux et Syst\`emes, CentraleSup\'elec, Univ. Paris Saclay, Gif sur Yvette, France.
$^4$Robert Bosch Center for Cyber-Physical Systems in Indian Institute of Science, Bangalore, India.
$^5$Division of Decision and Control Systems, KTH Royal Institute of Technology, Stockholm, Sweden.
$^6$Computer Science Department, University of Colorado Boulder, USA.
}
}

\makeatletter
\def\endthebibliography{%
	\def\@noitemerr{\@latex@warning{Empty `thebibliography' environment}}%
	\endlist
}
\makeatother

\begin{document}

\maketitle

\begin{abstract}
In this paper, we focus on the problem of compositional synthesis of controllers enforcing signal temporal logic (STL) tasks over a class of continuous-time nonlinear interconnected systems. 
By leveraging the idea of funnel-based control, we show that a fragment of STL specifications can be formulated as assume-guarantee contracts.  A new concept of contract satisfaction is then defined to establish our compositionality result, which allows us to guarantee the satisfaction of a global contract by the interconnected system when all subsystems satisfy their local contracts. Based on this compositional framework, we then design closed-form continuous-time feedback controllers to enforce local contracts over subsystems in a decentralized manner. Finally, we demonstrate the effectiveness of our results on two numerical examples.
\end{abstract}

\vspace{-0.3cm}
\section{Introduction}
In the last few decades, the world has witnessed rapid progresses in the development and deployment of cyber-physical systems (CPSs)\cite{baheti2011cyber}. Typical examples of real-world CPSs include smart grids and multi-robot systems.
Nowadays, these systems are often large-scale interconnected  resulting from tight interactions between computational components and physical entities, subjecting to complex specifications that are difficult to handle using classical control design approaches. 

To address the emerging challenges in dealing with modern CPSs, various approaches  \cite{tabuada2009verification,belta2017formal,ames2019control} have been developed to formally verify or synthesize certifiable controllers against rich specifications given by temporal logic formulae, such as linear temporal logics (LTL).  
Despite considerable development and progress in this field, when encountering large-scale CPSs, existing methods suffer severely from the \emph{curse of dimensionality}, which limits their applications to systems of moderate size.
To tackle this complexity issue, different compositional approaches have been proposed for the analysis and control of interconnected systems. 
The two most commonly used approaches are based on input-output properties (e.g., small-gain or dissipativity properties) \cite{rungger2016compositional,7857702,kim2017small} and assume-guarantee contracts \cite{kim2017small,saoud2021assume,sharf2021assume,7857702,saoud2020contract,al2020controller,chen2020safety, ghasemi2020compositional,shali2021behavioural}.
Both types of compositional approaches allow one to tackle large-scale complex systems in a divide and conquer manner, which considers a system as an interconnection of smaller subsystems, and breaks down complex large design or verification problems into sub-problems of manageable sizes.
The compositional framework proposed in the present paper will fall into the second category by leveraging assume-guarantee contracts (AGCs).
Specifically, the notion of assume-guarantee contracts (AGCs) prescribes properties that a component must guarantee under assumptions on the behavior of its environment (or its neighboring subsystems)\cite{benveniste2018contracts}. 

The main aim of this paper is to develop a compositional framework for the synthesis of controllers enforcing signal temporal logic (STL) formulae on continuous-time interconnected systems.
The control synthesis of STL properties for CPSs has attracted a lot of attentions in recent years.
Note that STL can be seen as an extension of LTL, which allows to formulate more expressive tasks with real-time and real-valued constraints \cite{maler2004monitoring}.
Moreover,  
unlike LTL formulae that are equipped with Boolean semantics (in which signals either satisfy or violate a formula), 
STL ones entail space robustness \cite{donze2010robust} 
which enables one to assess the robustness of satisfaction.
Despite the many advantages of STL formulae, 
the design of control systems under STL specifications is known to be a challenging task. 
In \cite{raman2014model}, 
the problem of synthesizing STL tasks on discrete-time systems is handled using model predictive control (MPC) where space robustness is encoded as mixed-integer linear programs.
The results in \cite{larsCDC} established a connection between funnel-based control and the robust semantics of STL specifications, based on which a continuous feedback control law is derived for continuous-time systems. This work is then extended to handle coupled multi-agent systems by providing a least violating solution for conflicting STL specifications \cite{lindemann2019feedback}.
The results in \cite{lindemann2018control} proposed for the first time a synthesis approach for STL tasks by leveraging a notion of time-varying control barrier functions. 



In this paper, we consider a fragment of STL specifications which is first formulated as funnel-based control problems.  By leveraging the derived funnels, we formalize the desired STL tasks as AGCs at the subsystem's level. 
A new concept of contract satisfaction, namely \emph{uniform strong satisfaction} (cf. Definition \ref{asg}), is introduced, which is critical for the compositional reasoning by making it possible to ensure the global satisfaction of STL tasks. 
Our main compositionality result is then presented using assume-guarantee reasoning, based on which the control of STL tasks can be conducted in a decentralized fashion. 
Finally, we derive continuous-time feedback controllers for subsystems in the spirit of funnel-based control, which ensures the satisfaction of local assume-guarantee contracts.
To the best of our knowledge, this paper is the first to handle STL specifications on continuous-time systems using assume-guarantee contracts.
Thanks to the derived closed-form control strategy and the decentralized framework, 
our approach requires very low computational complexity compared to existing results in the literature which mostly rely on discretizations in state space or time. 

{\bf{Related work:}} 
While AGCs have been extensively used in the computer science community~\cite{benveniste2018contracts,nuzzo2015compositional}, new frameworks of AGCs for dynamical systems with continuous state-variables have been proposed recently in~\cite{saoud2021assume,saoud2020contract} for continuous-time systems, and ~\cite{kim2017small},~\cite[Chapter 2]{saoud2019compositional} for discrete-time systems. In this paper, we follow the same behavioural framework of AGCs for continuous-time systems proposed in~\cite{saoud2021assume}. In the following, we provide a comparison with the approach proposed in~\cite{saoud2020contract,saoud2021assume}. A detailed comparison between the framework in~\cite{saoud2021assume}, the one in~\cite{kim2017small} and existing approaches from the computer science community~\cite{benveniste2018contracts,nuzzo2015compositional} can be found in~\cite[Section 1]{saoud2021assume}.

The contribution of the paper is twofold:
\begin{enumerate}[fullwidth,itemindent=0em]
    \item At the level of compositionality rules: The authors in~\cite{saoud2021assume} rely on a notion of \emph{strong contract satisfaction} to provide a compositionality result (i.e., how to go from the satisfaction of local contracts at the component's level to the satisfaction of the global specification for the interconnected system) under the condition of the set of guarantees (of the contracts) being closed. In this paper, we are dealing with STL specifications, which are encoded as AGCs made of open sets of assumptions and guarantees. The non-closedness of the set of guarantees makes the concept of contract satisfaction proposed in~\cite{saoud2021assume} not sufficient to establish a compositionality result. For this reason, in this paper, we introduce the concept of \emph{uniform strong contract satisfaction} and show how the proposed concept makes it possible to go from the local satisfaction of the contracts at the component's level to the satisfaction of the global STL specification at the interconnected system's level.
    \item At the level of controller synthesis: When the objective is to synthesize controllers to enforce the satisfaction of AGCs for continuous-time systems, to the best of our knowledge, existing approaches in the literature 
    can only deal with the particular class of invariance AGCs\footnote{where the set of assumptions and guarantees of the contract are described by invariants.} in \cite{saoud2020contract}, where the authors used symbolic control techniques to synthesize controllers. In this paper, we present a new approach to synthesize controllers for a more general class of AGCs, where the set of assumptions and guarantees are described by STL formulas, by leveraging tools in the spirit of funnel-based control.
\end{enumerate}






%
%
%
%
\vspace{-0.15cm}
\section{Preliminaries and Problem Formulation}\label{sec:pre}

{\bf{Notation:}} 
We denote by $\R$ and $\N$ the set of real and natural numbers, respectively. These symbols are annotated with subscripts to restrict them in the usual way, e.g., $\R_{>0}$ denotes the positive real numbers.
We denote by $\mathbb{R}^n$ an $n$-dimensional Euclidean space and by $\mathbb{R}^{n\times m}$ a space of real matrices with $n$ rows and $m$ columns. We denote by  $I_n$ the identity matrix of size $n$, and by $\mathbf{1}_n = [1,\dots,1]^\mathsf{T}$ the vector of all ones of size $n$. We denote by $\diag(a_1,\dots,a_n)$ the diagonal matrix with diagonal elements being $a_1,\dots,a_n$.

\subsection{Signal Temporal Logic (STL)} \label{STLsec}

Signal temporal logic (STL) is a predicate logic based on continuous-time signals, which consists of predicates $\mu$ that are obtained by evaluating a continuously differentiable predicate function $\mathcal{P}:\mathbb{R}^{n}\rightarrow \mathbb{R}$ as $\mu:=\begin{cases} \top & \text{if}\ \mathcal{P}(x) \geq 0\\ \perp & \text{if}\ \mathcal{P}(x) < 0, \end{cases}$ for $x \in \mathbb{R}^{n}$. 
The STL syntax is given by
\begin{equation*}
\phi::= \top \mid \mu \mid \neg\phi \mid \phi_{1}\wedge\phi_{2} \mid \phi_{1}\mathcal{U}_{[a, b]}\phi_{2},
\end{equation*}
where $\phi_{1}$, $\phi_{2}$ are STL formulae, and $\mathcal{U}_{[a, b]}$ denotes the temporal until-operator with time interval $[a, b]$, where $ a\leq b < \infty$. 
We use $(\mathbf{x}, t) \models\phi$ to denote that the state trajectory $\mathbf{x}: \mathbb{R}_{\geq 0}  \rightarrow X \subseteq \mathbb{R}^{n}$ satisfies $\phi$  at time $t$. 
The trajectory $\mathbf{x}: \mathbb{R}_{\geq 0}  \rightarrow X \subseteq \mathbb{R}^{n}$ satisfying formula $\phi$ is denoted by $(\mathbf{x}, 0) \models\phi$.
The semantics of STL \cite[Definition 1]{maler2004monitoring} can be recursively given by: $(\mathbf{x}, t)\models\mu$ if and only if $\mathcal{P}(\mathbf{x}(t))\geq 0, (\mathbf{x}, t)\models\neg\phi$ if and only if $\neg((\mathbf{x}, t)\models\phi), (\mathbf{x}, t)\models\phi_{1}\wedge\phi_{2}$ if and only if $(\mathbf{x}, t)\models \phi_{1}\wedge(\mathbf{x}, t)\models\phi_{2}$, and $(\mathbf{x}, t)\models\phi_{1}\mathcal{U}_{[a, b]}\phi_{2}$ if and only if $\exists t_{1}\in[t+a, t+b]$ s.t. $(\mathbf{x},\ t_{1})\models\phi_{2}\wedge\forall t_{2}\in[t, t_{1}]$, $(\mathbf{x}, t_{2})\models \phi_{1}$. Note that the disjunction-, eventually-, and always-operator can be derived as $\phi_{1}\vee\phi_{2}=\neg(\neg\phi_{1}\wedge\neg\phi_{2})$, $F_{[a, b]}\phi=\top\mathcal{U}_{[a, b]}{\phi}$, and $G_{[a, b]}\phi=\neg F_{[a, b]}{\neg\phi}$, respectively. 

Next, we introduce the robust semantics for STL (referred to as space robustness), which was originally presented in \cite[Definition 3]{donze2010robust}: 
$\rho^{\mu}(\mathbf{x}, t)$ := $\mathcal{P}(\mathbf{x}(t))$,
$\rho^{\neg\phi}(\mathbf{x}, t):  =-\rho^{\phi}(\mathbf{x}, t)$, 
$\rho^{\phi_{1}\wedge\phi_{2}}(\mathbf{x},  t): =\min(\rho^{\phi_{1}}(\mathbf{x},t), \rho^{\phi_{2}}(\mathbf{x}, t))$, 
$\rho^{F_{[a, b]}\phi}(\mathbf{x}, t): = \max_{t_{1}\in[t+a, t+b]}\rho^{\phi}(\mathbf{x}, t_{1})$,
$\rho^{G_{[a, b]}\phi}(\mathbf{x}, t):   = \min_{t_{1}\in[t+a, t+b]}\rho^{\phi}(\mathbf{x}, t_{1}).
$
Note that $(\mathbf{x}, t)\models\phi$ if $\rho^{\phi}(\mathbf{x}, t) > {0}$ holds \cite[Proposition 16]{fainekos2009robustness}. 
Space robustness determines how robustly a signal $\mathbf{x}$ satisfies the STL formula $\phi$. In particular, for two signals $\mathbf{x}_1,\mathbf{x}_2: \mathbb{R}_{\geq 0}  \rightarrow X$ satisfying a STL formula $\phi$ with $\rho^{\phi}(\mathbf{x}_2, t) > \rho^{\phi}(\mathbf{x}_2, t) > {0}$,  signal $\mathbf{x}_2$ is said to satisfy $\phi$ more robustly at time $t$ than $\mathbf{x}_1$ does.
We abuse the notation as $\rho^{\phi}(\mathbf{x}(t)):=\rho^{\phi}(\mathbf{x}, t)$ if $t$ is not explicitly contained in $\rho^{\phi}(\mathbf{x}, t)$. For instance, $ \rho^{\mu}(\mathbf{x}(t)): =\rho^{\mu}(\mathbf{x}, t): =\mathcal{P}(\mathbf{x}(t))$  since $\mathcal{P}(x(t))$ does not contain $t$ as an explicit argument. However, $t$ is explicitly contained in $\rho^{\phi}(\mathbf{x}, t)$ if temporal operators (eventually, always, or until) are used. 
Similarly as in \cite{aksaray2016q}, throughout the paper, the non-smooth conjunction is approximated by smooth functions as $\rho^{\phi_{1}\wedge\phi_{2}}(\mathbf{x}, t)\approx-\ln(\exp(-\rho^{\phi_{1}}(\mathbf{x}, t))+\exp(-\rho^{\phi_{2}}(\mathbf{x}, t)))$.

In the remainer of the paper, we will focus on a fragment of STL introduced above. Consider 
\begin{align}\label{syntax1}
\psi &::= \top \mid   \mu \mid \neg \mu \mid \psi_1 \wedge \psi_2, \\ \label{syntax2}
\phi &::= G_{[a,b]}\psi \mid F_{[a,b]} \psi  \mid F_{[\underline{a},\underline{b}]}G_{[\bar{a}, \bar{b}]}\psi,
\end{align}
where $\mu$ is the predicate, $\psi$ in \eqref{syntax2} and $\psi_1, \psi_2$ in \eqref{syntax1} are formulae of class $\psi$ given in \eqref{syntax1}. We refer to $\psi$ given in \eqref{syntax1} as non-temporal formulae, i.e., boolean formulae, while $\phi$ is referred to as (atomic) temporal formulae due to the use of always- and eventually-operators.

Note that this STL fragment allows us to encode concave temporal tasks, which is a necessary assumption used later for the design of closed-form, continuous feedback controllers (cf. Assumption \ref{assmprho1}). It should be mentioned that by leveraging the results in e.g., \cite{lindemann2020efficient}, it is possible to  expand our results to full STL semantics.  

\subsection{Interconnected control systems}

In this paper, we study the interconnection of finitely many continuous-time control subsystems. 
Consider a network consisting of $N \in \N$ control subsystems $\Sigma_i$, $i \in I = \{1,\dots,N\}$. 
For each $i \in I$, the set of \emph{in-neighbors} of $\Sigma_i$ is denoted by $\n_i \subseteq I \setminus \{i\}$, i.e., the set of subsystems $\Sigma_j$, $j \in \n_i$, directly influencing subsystem $\Sigma_i$.

A continuous-time control subsystem is formalized in the following definition.

\begin{definition} (Continuous-time control subsystem)
	\label{subsystem}
A continuous-time control subsystem $\Sigma_i$ is a tuple $\Sigma_i = (X_i,U_i,W_i,f_i,g_i,h_i)$, where 
\begin{enumerate} 
\item $X_i = \mathbb{R}^{n_i}$, $U_i =\mathbb{R}^{m_i}$ and $W_i = \mathbb{R}^{p_i}$ are the state, external input, and internal input spaces, respectively;
\item $f_i:\mathbb{R}^{n_i} \rightarrow \mathbb{R}^{n_i}$ is the flow drift, $g_i:\mathbb{R}^{n_i} \rightarrow \mathbb{R}^{n_i\times m_i}$ is the external input matrix, and $h_i: \mathbb{R}^{p_i} \rightarrow  \mathbb{R}^{n_i}$ is the internal input map. 
\end{enumerate} 
A trajectory of $\Sigma_i$ is an absolutely continuous map $(\mathbf{x}_i,\mathbf{u}_i,\mathbf{w}_i)\!:\!\mathbb{R}_{\geq 0} \!\rightarrow\! X_i \!\times\! U_i \!\times\!W_i$ such that 
for all $t \!\geq\! 0$ \vspace{-0.2cm}
\begin{equation}
\label{eqn:subsys}
    \mathbf{\dot x}_i(t)  = f_i(\mathbf{x}_i(t))+ g_i(\mathbf{x}_i(t))\mathbf{u}_i(t)+h_i(\mathbf{w}_i(t)),
\end{equation}
where $\mathbf{u}_i\!:\! \mathbb{R}_{\geq 0} \!\rightarrow\! U_i$ is the external input trajectory, and $\mathbf{w}_i :\!\mathbb{R}_{\geq 0} \rightarrow W_i$ is the internal input trajectory. 
\end{definition}

In the above definition, $w \!\in\! W$ are termed as ``internal" inputs describing the interaction between subsystems and $u \!\in\! U$ are ``external" inputs served as interfaces for controllers.

An interconnected control system is defined as follows.

\begin{definition}\label{intersys} (Continuous-time interconnected control system)
Consider $N \in \N$ control subsystems $\Sigma_i$ as in Definition \ref{subsystem}. An interconnected control system denoted by $\mathcal{I}(\Sigma_1,\dots,\Sigma_N)$ 
is a tuple $\Sigma = (X,U,f,g)$ where 
\begin{enumerate}
\item $X = \prod_{i\in I} X_i$ and $U = \prod_{i\in I} U_i$ are the state and external input spaces, respectively;
\item $f:\mathbb{R}^{n} \rightarrow \mathbb{R}^{n}$ is the flow drift and $g:\mathbb{R}^{n} \rightarrow \mathbb{R}^{n\times m}$ is the external input matrix defined as : 
$f(\mathbf{x}(t)) =  [f_1(\mathbf{x}_1(t))+h_1({\mathbf{w}_1(t)});
      \dots; f_N(\mathbf{x}_N(t))+h_N({\mathbf{w}_N(t)})]$, $
    g(\mathbf{x}(t)) =  \diag(g_1(\mathbf{x}_1(t)),$  $\dots, g_N(\mathbf{x}_N(t))),$ 
where $n=\sum_{i\in I}n_i$, $m=\sum_{i\in I}m_i$, $\mathbf{x} = [\mathbf{x}_1;\dots;\mathbf{x}_N]$,  $\mathbf{w}_i(t) =[\mathbf{x}_{j_1}(t);\dots;\mathbf{x}_{j_{|\n_i|}}(t)]$, for all $i \in I$.
\end{enumerate}
A trajectory of $\Sigma$ is an absolutely continuous map
$(\mathbf{x},\mathbf{u})\!:\!\mathbb{R}_{\geq 0} \!\rightarrow\! X \!\times\! U$
such that for all $t \geq 0$ \vspace{-0.15cm}
\begin{equation}
\label{eqn:sys}
    \mathbf{\dot x}(t)  = f(\mathbf{x}(t))+ g(\mathbf{x}(t))\mathbf{u}(t),
\end{equation}
where $\mathbf u: \mathbb{R}_{\geq 0} \rightarrow U$ is the external input trajectory.
\end{definition}

Note that in the above definition, the interconnection structure implies that all the internal inputs of a subsystem are states of its neighboring subsystems. Therefore, the definition of an interconnected control system boils down to the tuple $\Sigma = (X,U,f,g)$ since it has trivial null internal inputs.

We have now all the ingredients to provide a formal statement of the problem considered in the paper:

\begin{problem}
\label{pro:1}
Given an interconnected system $\Sigma\!=\!(X,$ $U,f,g)$, consisting of subsystems $\Sigma_i=(X_i,U_i,W_i,$ $f_i,g_i,h_i)$, $i \in \{1,\ldots,N\}$, and given an STL specification 
$\phi$ as in \eqref{syntax1}--\eqref{syntax2}, where $\phi=\land_{i=1}^N\phi_i$ and $\phi_i$ is the local STL task assigned to subsystem $\Sigma_i$, synthesize local controllers $\mathbf{u}_i:X_i\times \mathbb{R}_{\geq 0} \rightarrow U_i$ for subsystems $\Sigma_i$ such that $\Sigma$ satisfies the specification $\phi$. 
\end{problem}

In the remainder of the paper, to provide a solution to Problem \ref{pro:1},  
the desired STL tasks will be first casted as funnel functions in Section \ref{subsec: ppc}. Then, we present our main compositionality result based on a notion of assume-guarantee contracts as in Section \ref{subsec:agc}, which allows us to tackle the synthesis problem in a decentralized fashion.
We will further explain in Section \ref{subsec:stl_to_agc} on how to assign assume-guarantee contracts tailored to the funnel-based formulation of STL tasks.
A closed-form continuous-time control law will be derived in Section \ref{sec:localcontroller} to enforce local contracts over subsystems individually.  

\vspace{-0.15cm}
  

\section{Assume-Guarantee Contracts and Compositional Reasoning}\label{sec: agc}
In this section, we present a compositional approach based on a notion of assume-guarantee contracts, which enables us to reason about the properties of a continuous-time interconnected system based on the properties of its components.
Before introducing the compositionality result, in the next subsection, we first show how to cast STL formulae into  
time-varying funnel functions which will be leveraged later to design continuous-time AGCs. 
Note that the idea of casting STL as funnel functions was originally proposed in \cite{larsCDC}.  

\vspace{-0.15cm}
\subsection{Casting STL as funnel functions}\label{subsec: ppc}

First, let us define a funnel function
$\gamma_i(t) =  (\gamma_i^0 - \gamma_i^{\infty})\exp(-l_it) + \gamma_i^{\infty}$,   
where $l_i,t\in\R_{\geq 0}$, $\gamma_i^0,\gamma_i^{\infty} \in \R_{> 0}$ with $\gamma_i^0\geq \gamma_i^{\infty} $.
Consider the robust semantics of STL introduced in Subsection \ref{STLsec}. For each subsystem with STL specification $\phi_i$ in \eqref{syntax1} with the corresponding $\psi_i$, we can achieve  $0 < r_i \leq \rho_i^{\phi_i}(\mathbf{x}_i,0) \leq \rho_{i}^{max}$ by prescribing a temporal behavior to $\rho_i^{\psi_i}(\mathbf{x}_i(t))$ through a properly designed function $\gamma_i$ and parameter $\rho_{i}^{max}$, and the funnel \vspace{-0.15cm}
\begin{align} \notag
&-\gamma_i(t) + \rho_{i}^{max}  <  \rho_i^{\psi_i}(\mathbf{x}_i(t))  <  \rho_{i}^{max} \\     \label{goalineq}
& \iff  - \gamma_i(t) < \rho_i^{\psi_i}(\mathbf{x}_i(t)) -  \rho_{i}^{max}  < 0.
\end{align}
Note that functions $\gamma_i : \mathbb{R}_{\geq 0}  \rightarrow \mathbb{R}_{> 0}$, $i\in \{1,\ldots,N\}$, are positive, continuously differentiable, bounded, and non-increasing. 
The design of $\gamma_i$ and $\rho_{i}^{max}$ that leads to the satisfaction of $0 < r_i \leq \rho_i^{\phi_i}(\mathbf{x}_i,0) \leq \rho_{i}^{max}$  through \eqref{goalineq} will be discussed in Section \ref{Subsec:localcontroller}. 

To better illustrate the satisfaction of STL tasks using funnel-based strategy, we provide the next  example with more intuitions. 

\begin{example}

\begin{figure}[t!]
	\centering
	\subcaptionbox{Funnel $(-\gamma_1(t) + \rho_{1}^{max},  \rho_{1}^{max})$ (dashed lines) for $\phi_1:= F_{[0,8]}\psi_1$, s.t. $\rho_1^{\phi_1}(\mathbf{x},0) \geq r_1$ with $r_1 = 0.1$ (dotted line). \label{fig:funnel_f}}
	{\includegraphics[width=.45\textwidth]{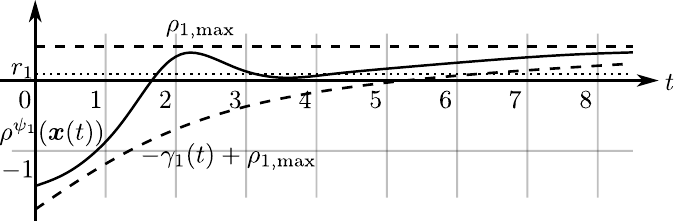}}
	\subcaptionbox{Funnel $(-\gamma_2(t) + \rho_{2}^{max}, \rho_{2}^{max})$  (dashed lines) for $\phi_2 := G_{[0,8]}\psi_2$, s.t. $ \rho_2^{\phi_2}(\mathbf{x},0) \geq r_2$ with $r_2 = 0.1$ (dotted line). \label{fig:funnel_g}}
	{\includegraphics[width=.45\textwidth]{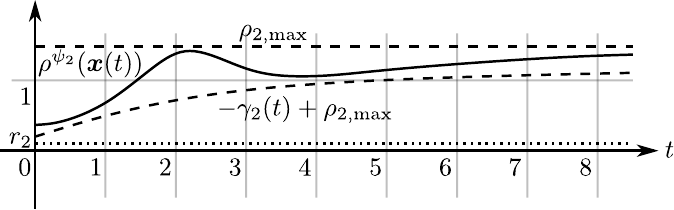}}
	\caption{Funnels for STL formulae. } 
	\label{fig:funnel_example}
	\vspace{-0.7cm}
\end{figure}
Consider STL formulae $\phi_1 := F_{[0,8]}\psi_1$ and $\phi_2 := G_{[0,8]}\psi_2$ with $\psi_1 = \mu_1$ and $\psi_2 = \mu_2$, where $\mu_1$ and $\mu_2$ are associated with predicate functions $\mathcal{P}_1(\mathbf{x}) = \mathcal{P}_2(\mathbf{x}) = \mathbf{x}$. 
Figs.~\ref{fig:funnel_f} and \ref{fig:funnel_g} show the funnel in \eqref{goalineq} prescribing a desired temporal behavior to satisfy $\phi_1$ and $\phi_2$, respectively. 
Specifically, it can be seen that
$\rho_1^{\psi_1}(\mathbf{x}(t)) \in (-\gamma_1(t) + \rho_{1}^{max},  \rho_{1}^{max})$ and $\rho_2^{\psi_1}(\mathbf{x}(t)) \in (-\gamma_2(t) + \rho_{2}^{max}, \rho_{2}^{max})$ for all $t \in \mathbb{R}_{\geq 0}$ as in Fig.~\ref{fig:funnel_example}. This shows that \eqref{goalineq} is satisfied for all $t \in \mathbb{R}_{\geq 0}$. 
Then, the connection between atomic formulae $\rho_i^{\psi_{i}}(\mathbf{x}(t))$ and temporal formulae $\rho_i^{\phi_i}(\mathbf{x},0)$, is made by the choice of  $\gamma_1$, $\gamma_2$, $\rho_{1}^{max}$, and $\rho_{2}^{max}$. For example,  the lower funnel $-\gamma_1(t) + \rho_{1}^{max}$ in Fig.~\ref{fig:funnel_f} ensures that $\rho_1^{\psi_1}(\mathbf{x}(t))  \geq r_1 = 0.1$ for all $t \geq 6$, which guarantees that the STL task $\phi_1$ is robustly satisfied by  $\rho_1^{\phi_1}(\mathbf{x},0) \geq r_1$.  
\end{example}



In the sequel, STL tasks will be formulated as contracts by leveraging the above-presented funnel-based framework. We will then design local controllers enforcing the local contracts over the subsystems (cf. Section \ref{sec:localcontroller}, Theorem \ref{theorem}).

\subsection{Compositional reasoning via assume-guarantee contracts}\label{subsec:agc}


In this subsection, we introduce a notion of continuous-time assume-guarantee contracts to establish our compositional framework. A new concept of contract satisfaction is defined which is tailored to the funnel-based formulation of STL specifications as discussed in Subsection \ref{subsec: ppc}.

%

\begin{definition} (Assume-guarantee contracts) \label{asg}
	Consider a subsystem
	$\Sigma_i \!=\! (X_i,U_i,W_i,f_i,g_i,h_i)$.  An assume-guarantee contract for $\Sigma_i$ is a tuple $\C_i \!=\! (A_i, G_i)$ where
	\begin{enumerate}
	\item $A_i: \mathbb{R}_{\geq 0}  \rightarrow  W_i$ is a set of assumptions on the internal input trajectories;
	\item $G_i: \mathbb{R}_{\geq 0}  \rightarrow  X_i$ is a set of guarantees on the state trajectories.
	\end{enumerate}
	
We say that $\Sigma_i$ (\emph{weakly}) satisfies $\mathcal{C}_i$, denoted by $\Sigma_i \models \mathcal{C}_i$, if for any trajectory  $(\mathbf{x}_i,\mathbf{u}_i,\mathbf{w}_i)\! : \! \mathbb{R}_{\geq 0}\! \rightarrow\! X_i \!\times\! U_i \!\times\! W_i$ of $\Sigma_i$, the following holds:
for all $t \!\in\! \mathbb{R}_{\geq 0} $ such that  $\mathbf{w}_i(s) \!\in\! A_i(s)$ for all $s \!\in\! [0,t]$, we have $\mathbf{x}_i(s) \!\in\! G_i(s)$ for all $s \!\in\! [0,t]$.

We say that $\Sigma_i$ \emph{uniformly strongly} satisfies $\mathcal{C}_i$, denoted by $\Sigma_i \models_{us} \mathcal{C}_i$, if for any trajectory $(\mathbf{x}_i,\mathbf{u}_i,\mathbf{w}_i) \!:\! \mathbb{R}_{\geq 0}  \!\rightarrow\! X_i \!\times\! U_i \!\times\! W_i$ of $\Sigma_i$, the following holds: 
there exists $\delta_i>0$ such that for all $t \in \mathbb{R}_{\geq 0} $ and for all $s \in [0,t]$ where $\mathbf{w}_i(s) \in A_i(s)$, we have $\mathbf{x}_i(s) \in G_i(s)$ for all $s \in [0,t+\delta_i]$.
\end{definition}

Note that $\Sigma_i \models_{us} \mathcal{C}_i$ obviously implies $\Sigma_i \models \mathcal{C}_i$. 

\begin{remark}
It should be mentioned that interconnected systems have no assumptions on internal inputs since they have a trivial null internal input set as in Definition \ref{intersys}. 
Hence, an AGC for an interconnected system $\Sigma = \mathcal{I}(\Sigma_1,\dots,\Sigma_N)$ will be denoted by $\mathcal{C}=(\emptyset, G)$.  
The concepts of contract satisfaction by $\Sigma$ are similar as in the above definition by removing the conditions on internal inputs.
\end{remark}




We are now ready to state the main result of this section providing conditions under which one can go from the satisfaction of local contracts at the subsystem's level to the satisfaction of a global contract for the interconnected system.

\begin{theorem}
\label{thm:main}
Consider an  interconnected control system  $\Sigma = \mathcal{I}(\Sigma_1,\dots,\Sigma_N)$ as in Definition \ref{intersys}.
To each subsystem $\Sigma_i$, $i \in I$, we associate a contract $\C_i = (A_{i}, G_{i})$ and let $\C  = (\emptyset, G)= (\emptyset,\prod_{i\in I} G_{i})$ be the corresponding contract for $\Sigma$. Assume the following conditions hold:
 \begin{enumerate}[leftmargin = 2em] 
 \item[(i)] For all $i \in I$ and for any trajectory $(\mathbf{x}_i,\mathbf{u}_i,\mathbf{w}_i):\mathbb{R}_{\geq 0} \rightarrow X_i\times U_i \times W_i$ of $\Sigma_i$,   $\mathbf{x}_i(0) \in G_i(0)$;
 \item[(ii)] for all $i \in I$, $\Sigma_i \models_{us} \C_i$;
 \item[(iii)] for all $i \in I$, $\prod_{j\in \n_i} G_{i} \subseteq A_{i}$.
 \end{enumerate}
 Then, $\Sigma \models \mathcal{C}$.
\end{theorem}

\begin{proof}
Let $(\mathbf{x},\mathbf{u})\!:\!\mathbb{R}_{\geq 0} \!\rightarrow\! X \!\times\! U$ be a trajectory of system $\Sigma$. 
Then, from the definition of interconnected systems, we have for all $i \in I$,  $(\mathbf{x}_i,\mathbf{u}_i,\mathbf{w}_i):\mathbb{R}_{\geq 0} \rightarrow X_i\times U_i \times W_i$ is a trajectory of $\Sigma_i$, where $\mathbf{w}_i(t)  =[\mathbf{x}_{j_1}(t);\dots;\mathbf{x}_{j_{|\n_i|}}(t)]$. Let us show the existence of $\delta>0$ such that for all $n \in \mathbb{N}$, $\mathbf{x}(s) \in G(s)$ for all $s \in [n\delta, (n+1)\delta)]$. We have from (i) that for all $i\in I$, $\mathbf{w}_i(0)\!=\![\mathbf{x}_{j_1}(0);\dots;\mathbf{x}_{j_{|\n_i|}}(0)] \in \prod_{j\in \n_i} G_{i}(0) \subseteq A_{i}(0)$, where the last inclusion follows from (iii). Hence, it follows from (ii) the existence of $\delta_i>0$, $i\in I$ such that $\mathbf{x}_i(s) \in G_i(s)$ for all $s \in [0,\delta_i]$. Let us define $\delta>0$ as $\delta \!:=\! \min_{i \in I} \delta_i$ and let us show by induction that $\mathbf{x}(s) \in G(s)$ for all $s \in [n\delta, (n+1)\delta)]$. First, one has from above that $\mathbf{x}_i(s) \in G_i(s)$ for all $s \in [0,\delta]$, which implies that $\mathbf{x}(s) \in G(s)$ for all $s \in [0,\delta]$. Let us now assume that $\mathbf{x}(s) \in G(s)$ for all $s \in [n\delta, (n+1)\delta)]$ and show that $\mathbf{x}(s) \in G(s)$ for all $s \!\in\! [(n+1)\delta, (n+2)\delta)]$. We have from the assumption that that for all $i\in I$, and for all $s \in [n\delta, (n+1)\delta]$, $\mathbf{w}_i(s)\!=\!$ $[\mathbf{x}_{j_1}(s);\dots;\mathbf{x}_{j_{|\n_i|}}(s)] \in \prod_{j\in \n_i}  G_{i}(s) \subseteq A_i(s)$,  where the last inclusion follows from (iii). Hence, from (ii) one gets  for all $i \in I$, $\mathbf{x}_i(s) \in G_i(s)$ for all $s \in [(n+1)\delta,(n+1)\delta+\delta_i]$, which further implies that $\mathbf{x}_i(s) \in G_i(s)$ for all $s \in [(n+1)\delta,(n+2)\delta]$ since $\delta \!:=\! \min_{i \in I} \delta_i$.  Hence, $\mathbf{x}(s) \in G(s)$ for all $s \in [(n+1)\delta, (n+2)\delta)]$, and one has that $\mathbf{x}(s) \in G(s)$ for all $s \geq 0$, which implies that $\Sigma \models \mathcal{C}$.
\end{proof}

\begin{remark}
It is important to note that while in the definition of the strong contract satisfaction in~\cite{saoud2021assume} the parameter $\delta$ may depend on time, our definition of assume-guarantee contracts requires a uniform $\delta$ for all time. The reason for this choice is that the uniformity of $\delta$ is critical in our compositional reasoning, since we do not require the set of guarantees to be closed as in~\cite{saoud2021assume} (See~\cite[Example 9]{saoud2021assume} for an example, showing that the compositionality result does not hold using the concept of strong satisfaction when the set of guarantees of the contract is open). Indeed, as it will be shown in the next section, the set of guarantees of the considered contracts are open and one will fail to provide a compositionality result based on the classical (non-uniform) notion of strong satisfaction in~\cite{saoud2021assume}.
\end{remark}

\subsection{From STL tasks to assume-guarantee contracts}\label{subsec:stl_to_agc}
The objective of the paper is to synthesize local controllers $\mathbf{u}_i:X_i\times \mathbb{R}_{\geq 0} \rightarrow U_i$, $i\in \{1,2,\ldots,N\}$, for subsystems $\Sigma_i$ to achieve the STL specification $\phi$, where $\phi=\land_{i=1}^N\phi_i$ and $\phi_i$ is the local STL task assigned to subsystem $\Sigma_i$. Hence, in view of the interconnection between the subsystems and the decentralized nature of the local controllers, one has to make some assumptions on the behaviour of the neighbouring components while synthesizing the local controllers. This property can be formalized in terms of contracts, where the contract should reflect the fact that the objective is to ensure that subsystem $\Sigma_i$ satisfies ``the guarantee" $\phi_i$ under ``the assumption" that each of the neighbouring subsystems $\Sigma_j$, satisfies its local task $\phi_j$, $j\in \mathcal{N}_i$. In this context, and using the concept of funnel function to cast the local STL tasks, $\phi_i$, $i\in \{1,2,\ldots,N\}$, presented in Section~\ref{subsec: ppc}, a natural assignment of the local assume-guarantee contract $\C_i = (A_{i}, G_{i})$ for the subsystems $\Sigma_i$ as in Definition \ref{subsystem}, can be defined formally as follows: 
\begin{enumerate}
\item $A_{i} 
= \prod_{j\in \n_i} \{ \mathbf{x}_{j}: \mathbb{R}_{\geq 0}  \rightarrow X_{j}  \mid  -\gamma_{j}(t)+\rho_{j}^{max} <  \rho_j^{\psi_{j}}(\mathbf{x}_{j}(t)) < \rho_{j}^{max}, \forall t \in \mathbb{R}_{\geq 0} \}$,
\item $G_{i} = \{\mathbf{x}_i : \mathbb{R}_{\geq 0}  \rightarrow X_i \mid -\gamma_i(t)+\rho_{i}^{max} < \rho_i^{\psi_{i}}(\mathbf{x}_i(t)) < \rho_{i}^{max}, \forall t \in \mathbb{R}_{\geq 0} \}$,
\end{enumerate}
where $\mathbf{x}_{j}$ denotes the state trajectories of the neighboring subsystem $\Sigma_j$, $j \!\in\! \n_i$, and $-\gamma_i, \rho_i^{\psi_i}, \rho_{i}^{max}$ are the functions discussed in Subsection \ref{subsec: ppc} corresponding to the STL task $\phi_i$. 

Once the specification $\phi$ is decomposed into local contracts\footnote{Note that the decomposition of a global STL formula is out of the scope of this paper. In this paper, we use a natural decomposition of the specification, where the assumptions of a component coincide with the guarantees of its neighbours. However, given a global STL for an interconnected system, one can utilize existing methods provided in recent literature, e.g., \cite{charitidou2021signal}, to decompose the global STL task into local ones.} and in view of Theorem~\ref{thm:main}, Problem~\ref{pro:1} can be resolved by considering local control problems for each subsystem $\Sigma_i$. These control problems can be solved in a decentralized manner and are formally defined as follows:

\begin{problem}
\label{pro:2}
Given a subsystem $\Sigma_i=(X_i,U_i,W_i,$ $f_i,g_i,h_i)$ and an assume-guarantee contract $\C_i=(A_i,G_i)$, where $A_i$ and $G_i$ are given by STL formulae by means of funnel functions, synthesize a local controller $\mathbf{u}_i:X_i\times \mathbb{R}_{\geq 0} \rightarrow U_i$ such that $\Sigma_i \models_{us} \C_i$.
\end{problem}

\section{Decentralized Controller Design}\label{sec:localcontroller}

In this section, we first provide a solution to Problem \ref{pro:2} by designing controllers ensuring that  local contracts for subsystems are uniformly strongly satisfied.
Then, we show that based on our compositionality result proposed in the last section, the global STL task for the network is satisfied by applying the derived local controllers to subsystems individually.

\subsection{Local controller design}\label{Subsec:localcontroller}

As discussed in Subsection \ref{subsec: ppc}, one can enforce STL tasks via funnel-based strategy by prescribing the temporal behavior of $\rho_i^{\psi_i}(\mathbf{x}_i(t))$ within the predefined region in \eqref{goalineq}, i.e., 
$$- \gamma_i(t) < \rho_i^{\psi_i}(\mathbf{x}_i(t)) -  \rho_{i}^{max}  < 0.$$
In order to design feedback controllers to achieve this, 
we translate the funnel functions into notions of errors as follows. 
First, define a one-dimensional error as  
$e_i(\mathbf{x}_i(t)) = \rho_i^{\psi_i}(\mathbf{x}_i(t))-\rho_{i}^{max}$.
Now, by normalizing the error $e_i(\mathbf{x}_i(t))$ with respect to the funnel function $\gamma_i$, we define the modulated error as 
$\hat e_i(\mathbf{x}_i,t)=\frac{e_i(\mathbf{x}_i(t))}{\gamma_i(t)}$.
Now, \eqref{goalineq} can be rewritten as 
$-1 < \hat e_i(t) < 0$. We use $\hat {\mathcal{D}_i} := (-1,0)$ 
to denote the performance region for $\hat e_i(t)$.
Next, the modulated error is transformed through a transformation function  $T_i: (-1,0) \rightarrow \mathbb{R}$ defined as
$$T_i(\hat e_i(\mathbf{x}_i,t)) = \ln(-\frac{\hat e_i(\mathbf{x}_i,t)+1}{\hat e_i(\mathbf{x}_i,t)}).$$
Note that the transformation function  $T_i: (-1,0) \rightarrow \mathbb{R}$ is a strictly increasing function, bijective and hence admitting an inverse. 
By differentiating the transformed error $\epsilon_i := T_i(\hat e_i(\mathbf{x}_i,t)) $ w.r.t time, we obtain \vspace{-0.15cm}
\begin{align}\label{transerror_dynamics}
\dot \epsilon_i  = \mathcal{J}_i(\hat e_i,t)[\dot e_i + \alpha_i(t)e_i],
\end{align}
where  $\mathcal{J}_i(\hat e_i,t)\!=\! \frac{\partial{T_i(\hat e_i)}}{\partial{\hat e_i}}\frac{1}{\gamma_i(t)} \!=\! -\frac{1}{\gamma_i(t)\hat e_i(1+ \hat  e_i)}\!>\!0$, for all $\hat e_i \!\in\! (-1,0)$, is the normalized Jacobian of the transformation function,
and $\alpha_i(t) \!= \!-\frac{\dot \gamma_i(t)}{\gamma_i(t)} \!>\!0$ for all $t \!\in \!\mathbb{R}_{\geq 0} $ is the normalized derivative of the performance function $\gamma_i$. 

It can be readily seen that, if the transformed error $\epsilon_i$ is bounded for all $t$, then the modulated error $\hat e_i$ is constrained within the performance region $\hat {\mathcal{D}_i}$, which further implies that the error $e_i$ evolves within the prescribed funnel bounds as in \eqref{goalineq}. 
Furthermore, we pose the following two assumptions on  functions $\rho_i^{\psi_i}$ for formulae $\psi_i$, which are required for the local controller design in the our main result of this section.

\begin{assumption}\label{assmprho1}
Each formula within class $\psi$ as in \eqref{syntax1} has the following properties: (i) $\rho_i^{\psi_i}: \mathbb{R}^{n_i} \rightarrow \mathbb{R}$ is concave and (ii) the formula is well-posed in the sense that for all $C \in \mathbb{R}$ there exists $\bar C \geq 0$ such that for all $\mathbf{x}_i \in  \mathbb{R}^{n_i}$ with $\rho_i^{\psi_i}(\mathbf{x}_i) \geq C$, one has $\Vert \mathbf{x}_i \Vert \leq \bar C < \infty$.
\end{assumption}

Define the global maximum of $\rho_i^{\psi_i}(\mathbf{x}_i)$ as $\rho_i^{opt}= \sup_{\mathbf{x}_i \in \mathbb{R}^{n_i}}\rho_i^{\psi_i}(\mathbf{x}_i)$. Note that $\psi_i$ is feasible only if $\rho_i^{opt} >0$, which leads to the following assumption.
\begin{assumption}\label{assmprho2}
The global maximum of $\rho_i^{\psi_i}(\mathbf{x}_i)$ is positive.	
\end{assumption}

The following assumption is required on subsystems in order to design  controllers   enforcing local contracts. 
\begin{assumption}\label{assmp:lipschitz}
Consider subsystem $\Sigma_i$ as in Definition \ref{subsystem}.
The functions $f_i : \mathbb{R}^{n_i} \rightarrow  \mathbb{R}^{n_i}$, $g_i  : \mathbb{R}^{n_i} \rightarrow  \mathbb{R}^{n_i \times m_i}$, and $h_i : \mathbb{R}^{p_i} \rightarrow  \mathbb{R}^{n_i}$ are locally Lipschitz continuous,  and $g_i(\mathbf{x}_i)g_i(\mathbf{x}_i)^\top$ is positive definite for all $\mathbf{x}_i \in \mathbb{R}^{n_i}$.
\end{assumption}

Now, we provide an important result in Proposition \ref{weaktostrong} to be used to prove the main theorem, which shows how to go from \emph{weak} to \emph{uniform strong satisfaction} of AGCs by relaxing the assumptions. 
The following notion of $\varepsilon$-closeness of trajectories is needed to measure the distance between continuous-time trajectories.

\begin{definition}(\cite{goedel2012hybrid}, $\varepsilon$-closeness of trajectories)
Let $Z \subseteq \mathbb{R}^n$.
Consider $\varepsilon > 0$ and two continuous-time trajectories $z_1 :  \mathbb{R}_{\geq 0}  \rightarrow Z $ and $z_2 :  \mathbb{R}_{\geq 0}  \rightarrow Z $. Trajectory $z_2$ is said to be $\varepsilon$-close to  $z_1$, if for all $t_1 \in \mathbb{R}_{\geq 0} $, there exists $t_2 \in \mathbb{R}_{\geq 0} $ such that $|t_1 - t_2| \leq \varepsilon$ and $\Vert z_1(t_1) -z_2(t_2) \Vert  \leq \varepsilon$. We define the $\varepsilon$-expansion of $z_1$ by : $\mathcal{B}_{\varepsilon}(z_1) = \{ z' : \mathbb{R}_{\geq 0}  \rightarrow Z \mid z' \text{ is } \varepsilon  \text{-close to } z_1\}$. For set $A = \{ z: \mathbb{R}_{\geq 0} \rightarrow Z\}$, $\mathcal{B}_{\varepsilon}(A) = \cup_{z\in A} \mathcal{B}_{\varepsilon}(z)$.
\end{definition}
 
Now, we introduce the following proposition which will be used later to prove our main theorem.

\begin{proposition}\label{weaktostrong}
(From weak to uniformly strong satisfaction of AGCs)
Consider a subsystem $\Sigma_i=(X_i,U_i,W_i,$ $f_i,g_i,h_i)$ associated with a local AGC $\C_i=(A_i,G_i)$. If trajectories of $\Sigma_i$ are uniformly continuous and  $\Sigma_i \models \C_i^{\varepsilon}$ with $ \C_i^{\varepsilon} = (\mathcal{B}_{\varepsilon}(A_i), G_i)$ for $\varepsilon>0$, then $\Sigma_i \models_{us} \C_i$.
\end{proposition}
\begin{proof}
    Consider $\varepsilon>0$ such that $\Sigma_i \models \C_i^{\varepsilon}$. From uniform continuity of $\mathbf{w}_i:\mathbb{R}_{\geq 0} \rightarrow W_i$ and for $\varepsilon>0$, we have the existence of $\delta>0$ such that for all $t \geq 0$, if $\mathbf{w}_i(s) \in A_i$, for all $s \in [0,t]$, then $\mathbf{w}_i(s) \in \mathcal{B}_{\varepsilon}(A_i)$, for all $s \in [0,t+\delta]$. Let us now show the uniform strong  satisfaction of contracts. Consider the $\delta >0$ defined above, consider $t \geq 0$ and assume that $\mathbf{w}_i(s) \in A_i(s)$ for all $s \in [0,t]$. Hence, we have from above that $\mathbf{w}_i(s) \in \mathcal{B}_{\varepsilon}(A_i)(s)$, for all $s \in [0,t+\delta]$, which implies from the weak satisfaction of the contract $\C_i^{\varepsilon}$ that $\mathbf{x}_i(s) \in G_i(s)$, for all $s \in [0,t+\delta]$. Hence, $\Sigma_i \models_{us} \C_i$.
\end{proof}

Now, we are ready to present the main result of this section solving Problem \ref{pro:2} for the local controller design.
\begin{theorem}\label{theorem}
Consider subsystem $\Sigma_i$ as in Definition \ref{subsystem} satisfying Assumption \ref{assmp:lipschitz}, with corresponding local assume-guarantee contract $\C_i = (A_{i}, G_{i})$, where
\begin{enumerate}
\item $A_{i} 
= \prod_{j\in \n_i} \{ \mathbf{x}_{j}: \mathbb{R}_{\geq 0}  \rightarrow X_{j}    \mid  -\gamma_{j}(t)+\rho_{j}^{max} < \rho_j^{\psi_{j}}(\mathbf{x}_{j}(t)) < \rho_{j}^{max}, \forall t \in \mathbb{R}_{\geq 0} \}$,
\item $G_{i} = \{\mathbf{x}_i : \mathbb{R}_{\geq 0}  \rightarrow X_i \mid -\gamma_i(t)+\rho_{i}^{max} < \rho_i^{\psi_{i}}(\mathbf{x}_i(t)) < \rho_{i}^{max}, \forall t \in \mathbb{R}_{\geq 0} \}$,
\end{enumerate}
where  $\psi_i$ is an atomic formula as in \eqref{syntax1} satisfying Assumptions \ref{assmprho1}-\ref{assmprho2}.
	If $-\gamma_i(0)\!+\!\rho_{i}^{max} \!<\! \rho_i^{\psi_i}(\mathbf{x}_i(0)) \!<\! \rho_{i}^{max} \!<\! \rho_i^{opt}$,	
	then the controller \vspace{-0.3cm}
	\begin{align} \notag
		\mathbf u_i(\mathbf{x}_i, t) =  &-g_i(\mathbf{x}_i)^\top\frac{\partial\rho_i^{\psi_i}(\mathbf{x}_i)^\top}{\partial \mathbf{x}_i} \mathcal{J}_i(\hat e_i,t)\epsilon_i(\mathbf{x}_i, t)  \\ \label{controller}
		& - g_i(\mathbf{x}_i)^\top 
		h_i(d_i(t))
	\end{align}\\[-20pt]
ensures that $\Sigma_i \!\!\models_{us} \!\! \mathcal{C}_i$,
where $d_i(t)\!\!\!\!=\!\!\![\gamma_{j_1}(t)\mathbf{1}_{n_{j_1}};\!\dots\!;$ $\gamma_{j_{|\n_i|}}(t)\mathbf{1}_{n_{|\n_i|}}]$.
\end{theorem}
\begin{proof}
We prove the uniform strong satisfaction of the contract using Proposition \ref{weaktostrong}. Let $(\mathbf{x}_i,\mathbf{u}_i,\mathbf{w}_i) : \mathbb{R}_{\geq 0}  \rightarrow X_i \times U_i \times W_i$ be a trajectory of $\Sigma_i$. Since $-\gamma_i(0)+\rho_{i}^{max} < \rho_i^{\psi_i}(\mathbf{x}_i(0)) < \rho_{i}^{max}$ holds, we have $\mathbf{x}_i(0) \in G_{i}$. Now, consider  $\varepsilon >0$. Let us prove $\Sigma_i \models \mathcal{C}_i^{\varepsilon}$, where $\mathcal{C}_i^{\varepsilon} = (\mathcal{B}_{\varepsilon}(A_{i}), G_{i})$. Let $s \in \mathbb{R}_{\geq 0} $, such that for all $t \in [0,s]$, $\mathbf{w}_i(t) \in \mathcal{B}_{\varepsilon}(A_{i})$, i.e., for all $j \in  \n_i $, 
$-\gamma_{j}(t)+\rho_{j}^{max}-\varepsilon < \rho_j^{\psi_{j}}(\mathbf{x}_{j}(t)) < \rho_{j}^{max}+  \varepsilon $ holds for all $t \in [0,s]$. Next, we show that ${\mathbf{x}_i}_{|[0,s]} \in G_{i}$.


Now, consider a Lyapunov-like function $V : \mathbb{R} \rightarrow \mathbb{R}_{\geq 0}$ defined as $V(\epsilon_i) = \frac{1}{2}\epsilon_i^2$. By differentiating $V$ with respect to time, we obtain 
\begin{align} \notag
 \dot V(t) =& \epsilon_i\dot \epsilon_i \stackrel{\eqref{transerror_dynamics}}= \epsilon_i\mathcal{J}_i(\hat e_i,t)[\dot e_i + \alpha_i(t)e_i] \\ \notag
 =& \epsilon_i\mathcal{J}_i(\hat e_i,t)[\frac{\partial{\rho_i^{\psi_i}(\mathbf{x}_i)}}{\partial \mathbf{x}_i}^\top \mathbf{\dot x}_i(t) - \frac{\dot \gamma_i(t)}{\gamma_i(t)}e_i] \\ \notag 
=& \epsilon_i\mathcal{J}_i(\hat e_i,t)\frac{\partial{\rho_i^{\psi_i}(\mathbf{x}_i)}}{\partial \mathbf{x}_i}^\top \big(f_i(\mathbf{x}_i)+ g_i(\mathbf{x}_i)\mathbf{u}_i(t)+h_i(\mathbf{w}_i)\big)\\ \label{dotV0}
&- \epsilon_i\mathcal{J}_i(\hat e_i,t)\dot \gamma_i(t)\hat e_i.
\end{align} 
\begin{figure*}[ht]
	\rule{\textwidth}{0.2pt}
	\small{
\begin{align}\notag
 \dot V(t) \stackrel{\eqref{controller}}= &\epsilon_i\mathcal{J}_i(\hat e_i,t)\frac{\partial{\rho_i^{\psi_i}(\mathbf{x}_i)}}{\partial \mathbf{x}_i}^\top \Big( f_i(\mathbf{x}_i)\!\!+\!\! g_i(\mathbf{x}_i)\big(\!\!-\!\!g_i(\mathbf{x}_i)^\top\frac{\partial\rho_i^{\psi_i}(\mathbf{x}_i)^\top}{\partial \mathbf{x}_i} \mathcal{J}_i(\hat e_i,t)\epsilon_i(\mathbf{x}_i, t)\!\! -\!\! g_i(\mathbf{x}_i)^\top h_i(d_i(t)) \big)\!\!  +\!\! h_i(\mathbf{w}_i(t))\Big)\!\!-\!\! \epsilon_i\mathcal{J}_i(\hat e_i,t)\dot \gamma_i(t)\hat e_i \\ \notag
=  &  -\epsilon_i\mathcal{J}_i(\hat e_i,t)\frac{\partial{\rho_i^{\psi_i}(\mathbf{x}_i)}}{\partial \mathbf{x}_i}^\top  g_i(\mathbf{x}_i)g_i(\mathbf{x}_i)^\top\frac{\partial\rho_i^{\psi_i}(\mathbf{x}_i)^\top }{\partial \mathbf{x}_i} \mathcal{J}_i(\hat e_i,t)\epsilon_i(\mathbf{x}_i, t) + \epsilon_i\mathcal{J}_i(\hat e_i,t)\Big(\frac{\partial{\rho_i^{\psi_i}(\mathbf{x}_i)}}{\partial \mathbf{x}_i}^\top  \big( f_i(\mathbf{x}_i)  +  h_i(\mathbf{w}_i) \\ \notag
&-  g_i(\mathbf{x}_i)g_i(\mathbf{x}_i)^\top h_i(d_i(t))\big) \!\!-\!\! \dot \gamma_i(t)\hat e_i \Big) \\ \notag
\leq  & -2(\lambda_{\min}(g_i(\mathbf{x}_i)g_i(\mathbf{x}_i)^\top)-\xi) \Vert \frac{\partial{\rho_i^{\psi_i}(\mathbf{x}_i)}}{\partial \mathbf{x}_i} \Vert^2 (\mathcal{J}_i(\hat e_i,t))^2 V - \xi \Vert \frac{\partial{\rho_i^{\psi_i}(\mathbf{x}_i)}}{\partial \mathbf{x}_i} \Vert^2(\epsilon_i\mathcal{J}_i(\hat e_i,t))^2  + \epsilon_i\mathcal{J}_i(\hat e_i,t)\Big(\frac{\partial{\rho_i^{\psi_i}(\mathbf{x}_i)}}{\partial \mathbf{x}_i}^\top \big( f_i(\mathbf{x}_i)  \\ \label{dotV} 
&+ h_i(\mathbf{w}_i) -g_i(\mathbf{x}_i)g_i(\mathbf{x}_i)^\top h_i(d_i(t)) \big) - \dot \gamma_i(t)\hat e_i \Big) \leq -\kappa V + \eta(t),
\end{align}
}
\rule{\textwidth}{0.2pt}
\vspace{-0.7cm}
\end{figure*}
Then, by substituting the control law \eqref{controller} in \eqref{dotV0}, for some  $0 < \xi < \lambda_{\min}(g_i(\mathbf{x}_i)g_i(\mathbf{x}_i)^\top)$, we get the chain of inequality as in \eqref{dotV},
where 
$\kappa =  2(\lambda_{\min}(g_i(\mathbf{x}_i)g_i(\mathbf{x}_i)^\top)-\xi) \min_{\mathbf{x}_i} (\Vert \frac{\partial{\rho_i^{\psi_i}(\mathbf{x}_i)}}{\partial \mathbf{x}_i} \Vert^2)
\frac{1}{\sup_{t \in \mathbb{R}_{\geq 0}  }\gamma_i(t)^2}\min_{\hat e_i \in \hat{\mathcal{D}_i}} (\frac{1}{\hat e_i(1+ \hat  e_i)})^2$, and
$\eta(t)\!\!  = \!\!  \frac{1}{\xi}\big( \Vert f_i(\mathbf{x}_i)\Vert ^2 \!+\!
\Vert h_i(\mathbf{w}_i)
\!- \!g_i(\mathbf{x}_i)g_i(\mathbf{x}_i)^\top h_i(d_i(t))\Vert^2\big)$ $+
\frac{\dot \gamma_i(t)^2\hat e_i^2}{2\xi\Vert \frac{\partial{\rho_i^{\psi_i}(\mathbf{x}_i)}}{\partial \mathbf{x}_i} \Vert^2}$.
Note that according to Assumption \ref{assmp:lipschitz} that  $g_i(\mathbf{x}_i)g_i(\mathbf{x}_i)^\top$ is positive definite, we know that $\lambda_{\min}(g_i(\mathbf{x}_i)g_i(\mathbf{x}_i)^\top) > 0$, hence a positive constant $\xi$ satisfying $0 < \xi < \lambda_{\min}(g_i(\mathbf{x}_i)g_i(\mathbf{x}_i)^\top)$ always exists. 

Now, we proceed with finding an upper bound  of $\eta(t)$, $t \in [0,s]$. 
Since $-\gamma_i(0)+\rho_{i}^{max} < \rho_i^{\psi_i}(\mathbf{x}_i(0)) < \rho_{i}^{max}$, $\mathbf{x}_i(0)$ is such that $\hat e_i(\mathbf{x}_i(0),0) \in \hat{\mathcal{D}_i} = (-1,0)$. Now, define the set $\mathcal{X}_i(t) := \{\mathbf{x}_i \in \mathbb{R}^{n_i}| -1< \hat e_i(\mathbf{x}_i,t) = \frac{\rho_i^{\psi_i}(\mathbf{x}_i)-\rho_{i}^{max}}{\gamma_i(t)} <0 \}$. Note that $\mathcal{X}_i(t)$ has the property that for $t_1<t_2$, $\mathcal{X}_i(t_2) \subseteq  \mathcal{X}_i(t_1)$ holds since $\gamma_i(t)$ is non-increasing in $t$.
Also note that $\mathcal{X}_i(t)$ is bounded due to condition (ii) of Assumption \ref{assmprho1} and $\gamma_i$ is bounded by definition, for all $i \in [1;N]$. According to  \cite[Proposition 1.4.4]{aubin2009set}, the inverse image of an open  set under a continuous function is open. By defining $\hat e_{i,0}(\mathbf{x}_i) := \hat e_i(\mathbf{x}_i, 0)$, we obtain that the inverse image $\hat e_{i,0}^{-1}(\hat {\mathcal{D}_i}) = \mathcal{X}_i(0)$ is open.
By the continuity of functions $f_i$, $g_i$ and $h_i$, 
it holds that for all states $\mathbf{x}_i \in \mathcal{X}_i(0)$, $\Vert f_i(\mathbf{x}_i)\Vert$ and $\Vert g_i(\mathbf{x}_i)g_i(\mathbf{x}_i)^\top h_i(d_i(t))\Vert$
are upper bounded, where 
$d_i(t)=[\gamma_{j_1}(t)\mathbf{1}_{n_{j_1}};\dots;\gamma_{j_{|\n_i|}}(t)\mathbf{1}_{n_{|\n_i|}}]$.
Note that by combining condition (ii) of Assumption \ref{assmprho1} and the assumption that $\mathbf{w}_i(t) \in \mathcal{B}_{\varepsilon}(A_{W_i})$, for all $t \in [0,s]$,  where $\mathbf{w}_i= [\mathbf{x}_{j_1};\dots;\mathbf{x}_{j_{|\n_i|}}]$,
 it holds that   $\Vert h_i(\mathbf{w}_i) \Vert$  is also upper bounded. 
Additionally, note that $\frac{\partial{\rho_i^{\psi_i}(\mathbf{x}_i)}}{\partial \mathbf{x}_i} = 0$ if and only if  $\rho_i^{\psi_i}(\mathbf{x}_i) = \rho_i^{opt}$ since  $\rho_i^\psi(\mathbf{x}_i)$ is concave under Assumption \ref{assmprho1}. 
However, since  
$ \rho_i^{\psi_i}(\mathbf{x}_i(0)) < \rho_{i}^{max} < \rho_i^{opt}$, and 
for all states $\mathbf{x}_i \in \mathcal{X}_i(0)$, $\rho_i^{\psi_i}(\mathbf{x}_i) < \rho_{i}^{max}$ holds,
then, we have for all states $\mathbf{x}_i \in \mathcal{X}_i(0)$, $\frac{\partial{\rho_i^{\psi_i}(\mathbf{x}_i)}}{\partial \mathbf{x}_i} \neq 0_{n_i}$, and $\Vert \frac{\partial{\rho_i^{\psi_i}(\mathbf{x}_i)}}{\partial \mathbf{x}_i} \Vert^2 \geq k_{\rho} >0$ holds for a positive constant $k_{\rho}$. Let us denote by $k_f \in  \mathbb{R}_{\geq 0} $, $k_h \in  \mathbb{R}_{\geq 0} $, and $k_g \in  \mathbb{R}_{\geq 0} $ the upper bounds satisfying $\max_{\mathbf{x}_i \in \mathcal{X}_i(0)} \Vert f_i(\mathbf{x}_i)\Vert \leq k_f$, $\max_{\mathbf{w}_i \in \mathcal{B}_{\varepsilon}(A_{W_i})}$ $\Vert h_i(\mathbf{w}_i)\Vert \leq k_h$, and $\max_{\mathbf{x}_i \in \mathcal{X}_i(0)} \Vert g_i(\mathbf{x}_i)g_i(\mathbf{x}_i)^\top h_i(d_i(t))\Vert \leq k_g$, respectively. 
Consequently, we can define an upper bound $\bar \eta$ of $\eta(t)$, $\forall \mathbf{x}_i \in \mathcal{X}_i(0)$, $\forall t \in \mathbb{R}_{\geq 0} $, as $\bar \eta =  \frac{k_f^2+k_h^2+k_g^2}{\xi} + \frac{|\dot \gamma_i(0)|^2}{2\xi k_{\rho}}$, where $|\dot \gamma_i(0)|$ is bounded by definition.

Next, we show that $\hat{\mathcal{D}_i}$ is an attraction set. 
To do this, we first introduce a function $\mathcal{S}(\hat{e}_i) = 1- e^{-V(\hat{e}_i)}$ for which $0 < \mathcal{S}(\hat{e}_i) <1$, $\forall \hat{e}_i \in \hat {\mathcal{D}_i}$ and $\mathcal{S}(\hat{e}_i) \rightarrow 1$ as $\hat{e}_i \rightarrow \partial \hat {\mathcal{D}_i}$. By differentiating $\mathcal{S}(\hat{e}_i)$ we get 
\begin{align}\label{dotmV}
\dot {\mathcal{S}}(t) = \dot V(\hat{e}_i)\big( 1-\mathcal{S}(\hat{e}_i)\big) .
\end{align}
By substituting \eqref{dotV} and inserting $V(\hat{e}_i) = - \ln{(1-\mathcal{S}(\hat{e}_i))}$ in \eqref{dotmV}, we get
\begin{align}\label{dotmVb}
\hspace{-0.1cm}
\dot {\mathcal{S}}(t) \leq -\kappa \big(1-\mathcal{S}(\hat{e}_i) \big)\big(\ln{( e^{-\frac{\eta(t)}{\kappa}})}- \ln{(1-\mathcal{S}(\hat{e}_i))}\big).
\end{align}
Note that by definition, we have $\kappa>0$ and $1-\mathcal{S}(\hat{e}_i)> 0$. 
Now define the region $\Omega_{\hat e_i} = \{\hat e_i \in \hat{\mathcal{D}_i} | \mathcal{S}(\hat e_i) \leq 1- e^{-\frac{\bar \eta}{\kappa}}\}$. 
Since $-\gamma_i(0)+\rho_{i}^{max} < \rho_i^{\psi_i}(\mathbf{x}_i(0)) < \rho_{i}^{max}$ holds, then we obtain that $\hat e_i(\mathbf{x}_i(0)) \in \hat{\mathcal{D}_i} = (-1,0)$,
and consequently, $\mathcal{S}(\hat{e}_i(0)) <1$ holds. Let us define $c = \mathcal{S}(\hat{e}_i(0))$ and the set $\Omega_c = \{ \hat e_i \in \hat{\mathcal{D}_i} | \mathcal{S}(\hat e_i) \leq c \}$. Now, consider the case when $c < 1- e^{-\frac{\bar \eta}{\kappa}}$. In this case, $\Omega_c \subset \Omega_{\hat e_i}$, and by \eqref{dotmVb}, $\dot {\mathcal{S}}(t) \leq 0$ for all  $\hat e_i \in \partial{\Omega_{\hat e_i}}$, therefore, $\hat e_i(t) \in \Omega_{\hat e_i}$, $\forall t \in \mathbb{R}_{\geq 0} $. Next, consider the other case when $c \geq 1- e^{-\frac{\bar \eta}{\kappa}}$. In this case $\Omega_{\hat e_i} \subseteq \Omega_c$, and by \eqref{dotmVb}, $\dot {\mathcal{S}}(t) < 0$ for all  $\hat e_i \in \Omega_c \setminus \Omega_{\hat e_i}$, hence, $\mathcal{S}(\hat e_i) \rightarrow \Omega_{\hat e_i}$. Thus, starting from any point within the set $\Omega_c$, $\mathcal{S}(\hat e_i(t))$ remains less than $1$. Consequently, the modulated error $\hat e_i$ always evolves within a closed strict subset of  $\hat{\mathcal{D}_i}$  (that is, set $\Omega_{\hat e_i}$ in the case that $c < 1- e^{-\frac{\bar \eta}{\kappa}}$,  or set $\Omega_c$ in the case that $c \geq 1- e^{-\frac{\bar \eta}{\kappa}}$), which implies that $\hat e_i$ is not approaching the boundary $\partial \hat {\mathcal{D}_i}$. It follows that  the transformed error $\epsilon_i$ is bounded. Thus, we can conclude that  $\rho_i^{\psi_i}(\mathbf{x}_i(t))$ evolves within the predefined region \eqref{goalineq}, i.e., ${\mathbf{x}_i}_{|[0,s]} \in G_{i}$. 
Therefore, we have $\Sigma_i \models \mathcal{C}_i^{\varepsilon}$. By Proposition \ref{weaktostrong}, it implies that $\Sigma_i \models_{us} \mathcal{C}_i$.
\end{proof}

The proof of Theorem \ref{theorem} was partly inspired by the proof of \cite[Thm. 1]{karayiannidis2012multi}, where similar Lyapunov arguments were used in the context of funnel-based consensus control of multi-agent systems.

Remark that the connection between atomic formulae $\rho_i^{\psi_{i}}(\mathbf{x}_i(t))$ and temporal formulae $ \rho_i^{\phi_i}(\mathbf{x}_i,0)$ is made by $\gamma_i$ and $\rho_{i}^{max}$ as in \eqref{goalineq}, which need to be designed as instructed in \cite{larsCDC}. 
Specifically, 
if Assumption \ref{assmprho2} holds, select  \vspace{-0.2cm}


\begin{align} \label{funnelpara1}
\small
 t_{i}^{\ast}\in & \begin{cases} a_{i} & \text{if}\ \phi_{i}=G_{[a_{i}, b_{i}]}\psi_{i}\\ 
		{[}a_{i}, b_{i}] & \text{if}\ \phi_{i}=F_{[a_{i}, b_{i}]}\psi_{i}\\ 
		{[}\underline{a}_{i}+\bar{a}_{i}, \underline{b}_{i} +  \bar{a}_{i}] \! & \text{if}\ \phi_{i}=F_{[\underline{a}_{i},\underline{b}_{i}]}G_{[\bar{a}_{i}, \bar{b}_{i}]}\psi_{i}  \end{cases} \\ 
	\rho_{i}^{\max}\in & \left(\max(0, \rho_i^{\psi_{i}}(\mathbf{x}_{{i}}(0))), \rho_{{i}}^{\text{opt}}\right) \\ 
	r_{i}\in & (0,\rho_{i}^{\max}) 
	\end{align} 
	\begin{align}
\small
	\gamma_{i}^{0}\in&\begin{cases} (\rho_{i}^{\max}-\rho_i^{\psi_{i}}(\mathbf{x}_{{i}}(0)),\infty)\ & \text{if}\ t_{i}^{\ast} > 0\\ (\rho_{i}^{\max}-\rho_i^{\psi_{i}}(\mathbf{x}_{{i}}(0)),\rho_{i}^{\max}-r_{i}]\ & \text{else}  \end{cases} \\ \gamma_{i}^{\infty}\in&\left.\left(0, \min(\gamma_{i}^{0}, \rho_{i}^{\max}-r_{i})\right]\right. \\ l_{i}\in&\begin{cases}\mathbb{R}_{\geq 0} & \text{if}- \gamma_{i}^{0}+\rho_{i}^{\max}\geq r_{i}\\   \label{funnelpara2}
	\frac{-\ln\left(\frac{r_{i}+\gamma_{i}^{\infty}-\rho_{i}^{\max}}{-\gamma_{i}^{0}-\gamma_{i}^{\infty}}\right)}{t_{i}^{\ast}}& \text{else}. \end{cases} 
\end{align} 	
With $\gamma_i$ and $\rho_{i}^{max}$ chosen properly (as shown above), one can achieve $0 \!<\! r_i \!\leq\! \rho_i^{\phi_i}(\mathbf{x}_i,0) \!\leq\! \rho_{i}^{max}$ by prescribing a temporal behavior to $\rho_i^{\psi_{i}}(\mathbf{x}_i(t))$ as in the set of guarantee $G_{i}$  in Theorem \ref{theorem}, i.e., $  -\gamma_i(t)\!+\!\rho_{i}^{max} \!<\! \rho_i^{\psi_{i}}(\mathbf{x}_i(t)) \!< \!\rho_{i}^{max}$ for all $t \!\geq\! 0$.

\vspace{-0.1cm}
\subsection{Global task satisfaction}

In this subsection, we show that by applying the local controllers to the subsystems,
the global STL task for the network is also  satisfied based on our compositionality result.

\begin{corollary}
 	Consider an interconnected control system  $\Sigma = \mathcal{I}(\Sigma_1,\dots,\Sigma_N)$ as  in Definition \ref{intersys}. 
	If we apply the controllers as   in \eqref{controller} to  subsystems $\Sigma_i$, then we get $\Sigma \models \C  = (\emptyset,\prod_{i\in I} G_{i})$. This means that the control objective in Problem~\ref{pro:1} is achieved, i.e.,  system $\Sigma$ satisfies signal temporal logic task $\phi$.
\end{corollary}
\begin{proof}
	From Theorem \ref{theorem}, one can verify that the closed-loop subsystems under controller \eqref{controller} satisfy: 
	for all $i \in I$, $\Sigma_i \models_{us} \C_i$, and  for all $i \in I$, $\prod_{j\in \n_i} G_{i} \subseteq A_{i}$.
	Moreover, for all $i \in I$ and for any trajectory $(\mathbf{x}_i,\mathbf{u}_i,\mathbf{w}_i):\mathbb{R}_{\geq 0} \rightarrow X_i\times U_i \times W_i$ of $\Sigma_i$, the choice of parameters of the funnel as in \eqref{funnelpara1}--\eqref{funnelpara2} ensures that  $\mathbf{x}_i(0) \in G_i(0)$.
Hence, all conditions required in Theorem \ref{theorem} are satisfied, and thus,
we conclude that $\Sigma \models \C  = (\emptyset,\prod_{i\in I} G_{i})$ as a consequence of Theorem \ref{thm:main}. Therefore, the interconnection satisfies the STL task  $\phi=\land_{i=1}^N\phi_i$.
\end{proof}

\vspace{-0.15cm}
\section{Case Study}


\begin{figure}[t]
	\centering
	\includegraphics[width=.3\textwidth]{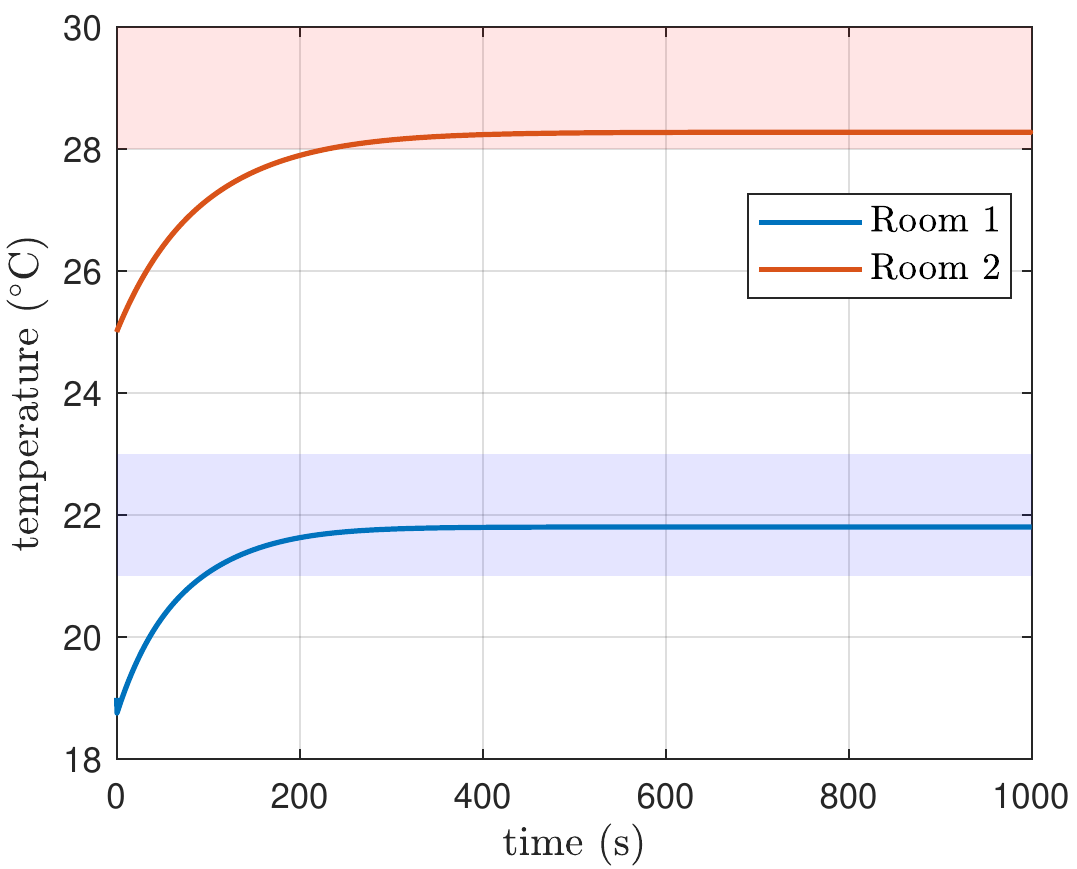}
 	\caption{Temperature evolution of the closed-loop subsystems $\Sigma_1$ and $\Sigma_2$ under control policy in \eqref{controller}.} \label{fig:room_temperature}
 	\vspace{-0.2cm}
\end{figure}

\begin{figure}[t]
	\vspace{+ 0.2cm}
	\centering
	\subcaptionbox{Funnel for $\Sigma_1$ with task $\psi_1$ \label{fig:room_rho2}}
	{\includegraphics[width=.238\textwidth]{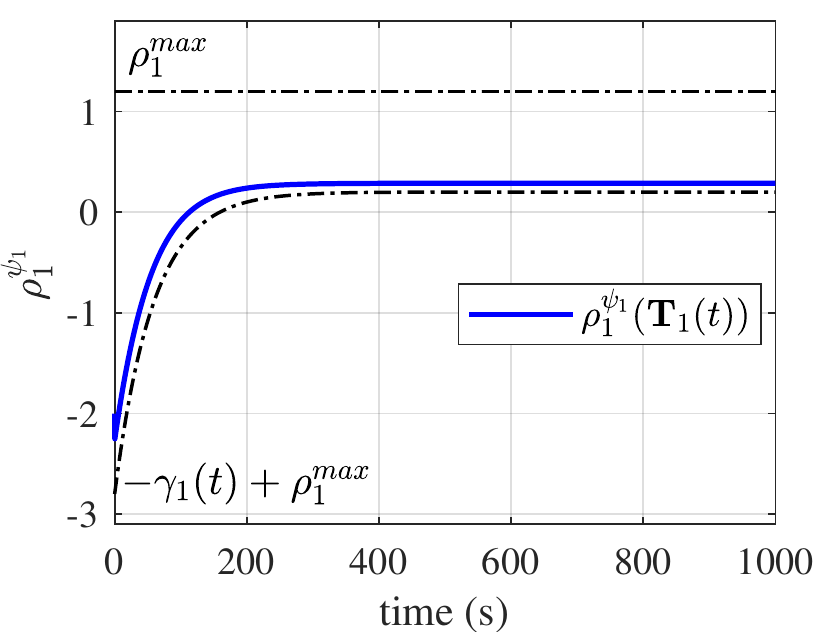}}
	\subcaptionbox{Funnel for $\Sigma_2$ with task $\psi_2$ \label{fig:room_rho4}}
	{\includegraphics[width=.238\textwidth]{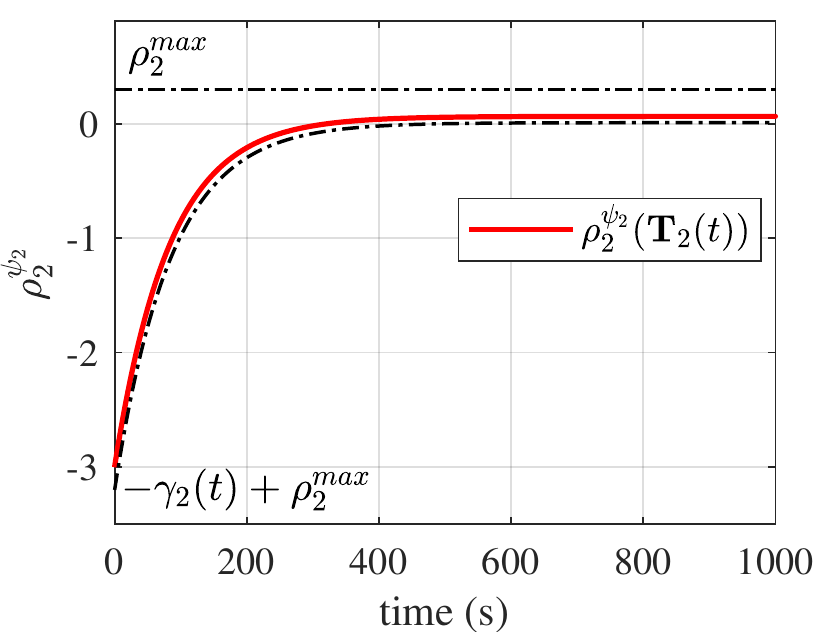}}
	\caption{Funnels for the local STL tasks for subsystems $\Sigma_1$ and $\Sigma_2$. Performance bounds are indicated by dashed lines. Evolution of $\rho_i^{\psi_{i}}(\mathbf{T}_i(t))$ are depicted using solid lines.}
	\label{fig:room}
	\vspace{-0.2cm}
\end{figure}

We demonstrate the effectiveness of the proposed results on two case studies: a room temperature regulation and a mobile robot control problem.

\subsection{Room Temperature Regulation} 

Here, we apply our results to the temperature regulation of a circular building with $N \geq 3$ rooms each equipped with a heater. 
 The evolution of the temperature of the interconnected model is described by the differential equation: 
\begin{align}\label{room}
\Sigma:\left\{
\begin{array}{rl}
\dot{\mathbf{T}}(t)=& A\mathbf{T}(t) + \alpha_h T_h \nu(t) + \alpha_e T_e,\\
\mathbf{y}(t)=&\mathbf{T}(t),
\end{array}
\right.
\end{align}
adapted from \cite{girard2015safety}, where $A \!\in\! \mathbb{R}^{N\!\times\! N}$ is a matrix with elements $\{A\}_{ii} \!=\! (- \!2 \alpha\!-\!\alpha_e\!-\!\alpha_h \nu_i)$, $\{A\}_{i,i+1} \!=\! \{A\}_{i+1,i} \! = \{A\}_{1,N} \!=\! \{A\}_{N,1} =\!\alpha$, $\forall i \in \{1, \dots, N-1\}$, and all other elements are identically zero, 
$\mathbf{T}(k)\!=\![\mathbf{T}_1(k);\dots; \mathbf{T}_N(k)]$,   $T_e\!=\![T_{e1};\dots;T_{eN}]$, $\nu(k)\!=\![\nu_1(k);\dots;\nu_N(k)]$, where $\nu_i(k)\!\in\! [0,1]$, $\forall i\!\in\!\{1, \dots, N\}$, represents the ratio of the heater valve being open in room $i$.
Parameters $\alpha\! =\! 0.05$, $\alpha_e\!=\!0.008$, and $\alpha_h\!=\! 0.0036$ are heat exchange coefficients, 
$T_{ei}\! =\! -1\,^\circ C$ is the external environment temperature, and  $T_h \!=\!50\,^\circ C$ is the heater temperature. 

Now, by introducing $\Sigma_i$ described by 
\begin{align*}\label{roomi}
\Sigma_i:\left\{
\begin{array}{rl}
\dot{\mathbf{T}}_i(t)=& a\mathbf{T}_i(t) + d \mathbf{w}_i(t)+ \alpha_h T_h \nu_i(t) + \alpha_e T_{ei},\\
\mathbf{y}_i(t)=&\mathbf{T}_i(t),
\end{array}
\right.
\end{align*}
where $a = - \!2 \alpha\!-\!\alpha_e\!-\!\alpha_h \nu_i$, $d = \alpha$, and $\mathbf{w}_i(t) = [\mathbf{y}_{i-1}(t);\mathbf{y}_{i+1}(t)]$ (with $\mathbf{y}_{0} = \mathbf{y}_{n}$ and $\mathbf{y}_{n+1} = \mathbf{y}_{1}$), one can readily verify that $\Sigma = \mathcal{I}(\Sigma_1,\dots,\Sigma_N)$ as in Definition \ref{intersys}. 
The initial temperatures of these rooms are, respectively, $\mathbf{T}_{i}(0)=19\,^\circ C$ if $i \in I_o = \{i \text{ is odd } | i \in \{1, \dots, N\} \}$, and $\mathbf{T}_{i}(0)=25\,^\circ C$ if $i \in I_e = \{i \text{ is even } | i \in \{1, \dots, N\} \}$. 
The room temperatures are subject to the following STL tasks 
$\phi_i$: $F_{[0,1000]}G_{[200,1000]}( \mathbf{T}_i\leq 25) \wedge(\mathbf{T}_i \geq 21)$, for $i \in I_o$, and $\phi_i$: $F_{[0,1000]}G_{[500,1000]}( \mathbf{T}_i\leq 30) \wedge(\mathbf{T}_i \geq 28)$, for $i \in I_e$. Intuitively, the STL tasks  $\phi_i$ requires that the controller (heater) should be synthesized such that the temperature of the first room reaches the specified region ($[21, 25]$ for odd-numbered rooms or $[28, 30]$ for the even-numbered room) and remains there in the desired time slots. 

Next, we apply the proposed funnel-based feedback controllers as in \eqref{controller} to enforce the STL tasks on consisting of $N = 1000$ rooms. 
Numerical implementations were performed using MATLAB on a computer with a processor Intel Core i7 3.6 GHz CPU. 
Note that the computation of local controllers took on average 0.01 ms, which is negligible. The computation cost is very cheap since the local controller $\boldsymbol{u}_{i}$ is given by a closed-form expression and computed individually for the subsystems only. 
The simulation results for subsystems $\Sigma_1$ and $\Sigma_2$ are shown in Figs.~\ref{fig:room_temperature} and \ref{fig:room}. The state trajectories of the closed-loop subsystems are depicted as in Fig.~\ref{fig:room_temperature}.  The shaded areas represent the desired temperature regions to be reached by the systems.  
In Fig.~\ref{fig:room}, we present the temporal behaviors of $\rho_i^{\psi_{i}}(\mathbf{T}_1(t))$ for the two rooms $\Sigma_1$ and $\Sigma_2$.
It can be readily seen that the prescribed performances of   $\rho_i^{\psi_{i}}(\mathbf{T}_i(t))$ are satisfied with respect to the error funnels, which shows that the time bounds are also respected.
Remark that the design parameters of the funnels are chosen according to the instructions listed in \eqref{funnelpara1}-\eqref{funnelpara2}, which guarantees the satisfaction of temporal formulae $ \rho_i^{\phi_i}(\mathbf{T}_i,0)$ by prescribing temporal behaviors of atomic formulae $\rho_i^{\psi_{i}}(\mathbf{T}_i(t))$ as in Fig.~\ref{fig:room}. 
We can conclude that all STL tasks are satisfied within the desired time interval. 

\subsection{Mobile Robot Control}

In this subsection, we demonstrate the effectiveness of the proposed results on a network of $N = 5$ mobile robots adapted from \cite{liu2008omni} with induced dynamical couplings.
Each mobile robot has three omni-directional wheels.
The dynamics of each robot $\Sigma_i$, $i\in\{1,2,\ldots,5\}$ can be described by 
\begin{equation*}
	\small
	\dot{\boldsymbol{x}}_{{i}} \!=\!\begin{bmatrix}
		\cos (x_{i, 3}) \! &\! -\sin (x_{i, 3}) \!&\! 0\\
		\sin (x_{i, 3}) \!&\! \cos (x_{i, 3}) \!&\! 0\\
		0 \!&\! 0 \!&\! 1
	\end{bmatrix}\left(B_{i}^{\top}\right)^{-1}\!\!R_{i}\boldsymbol{u}_{i} \!-\!\! \sum_{j\in \n_i }k_i(\boldsymbol{x}_{i} \!-\!  \boldsymbol{x}_{j}),
\end{equation*}
where the state variable of each robot is defined as $\boldsymbol{x}_{i}:=[x_{i, 1}; x_{i, 2};x_{i, 3}]$ with two states $x_{i, 1}$ and $x_{i, 2}$ indicating the robot position and  state $x_{i, 3}$ indicating the robot orientation with respect to the $x_{i, 1}$-axis; $R_{i}: = 0.02$ m is the wheel radius of each robot; $B_{i}: = \begin{bmatrix}0 & \cos(\pi/6) & -\cos(\pi/6)\\ -1 & \sin(\pi/6) & -\sin(\pi/6)\\ L_{i} & L_{i} & L_{i}\end{bmatrix}$ describes geometrical constraints with $L_{i}:=0.2$ m being the radius of the robot body. Each element of the input vector $\boldsymbol{u}_{i}$ corresponds to the angular rate of one wheel. 
Note that $\sum_{j\in \n_i }k_i(\boldsymbol{x}_{i} - \boldsymbol{x}_{j})$ represents the dynamical coupling between subsystems induced by an implemented consensus protocol, where $\boldsymbol{x}_{j}$, $j \in \n_i$, are the states of the neighboring subsystems of $\Sigma_i$, and $k_i = 0.1$. 
Specifically, we have $\n_i = \{i + 1\}$ for subsystems $\Sigma_i$, $i \in \{1,2,3,4\}$, 
and $\n_5 = \{1\}$. 
The initial states of the robots are, respectively, $\boldsymbol{x}_{1}(0)=[0.1; 0.6; \pi/4]$, $\boldsymbol{x}_{2}(0)=[0.4; 1.1; -\pi/4]$, $\boldsymbol{x}_{3}(0)=[1.05; 0.8; -\pi/4]$, $\boldsymbol{x}_{4}(0)=[1; 0.2; \pi/4]$,
$\boldsymbol{x}_{5}(0)=[0.3; 0.1; 0]$.

\begin{figure}[t]

	\centering
	\includegraphics[width=.35\textwidth]{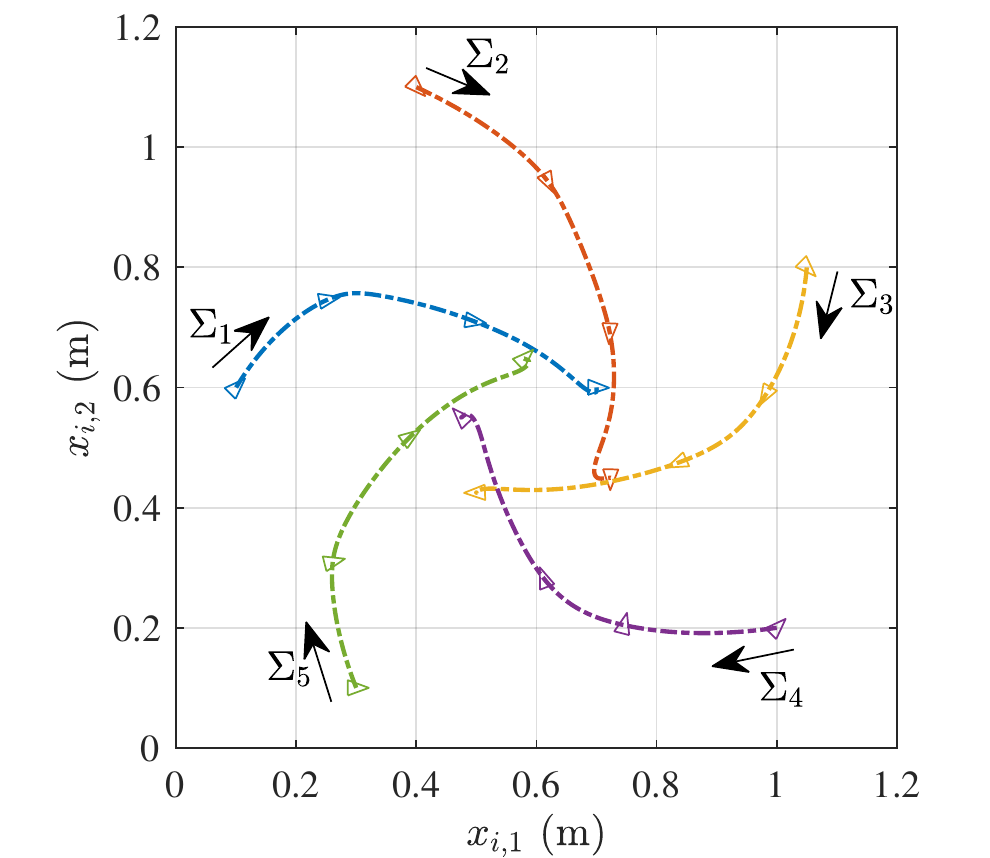}
 	\caption{State trajectories of the closed-loop robot systems on the position plane. The triangles indicate the orientation of each robot.} \label{fig:robot_position}
 	\vspace{-0.5cm}
\end{figure}

\begin{figure}[t!]
	\vspace{+ 0.2cm}
	\centering
	\subcaptionbox{Funnel for $\Sigma_2$ with task $\psi_2$ \label{fig:robot_rho2}}
	{\includegraphics[width=.238\textwidth]{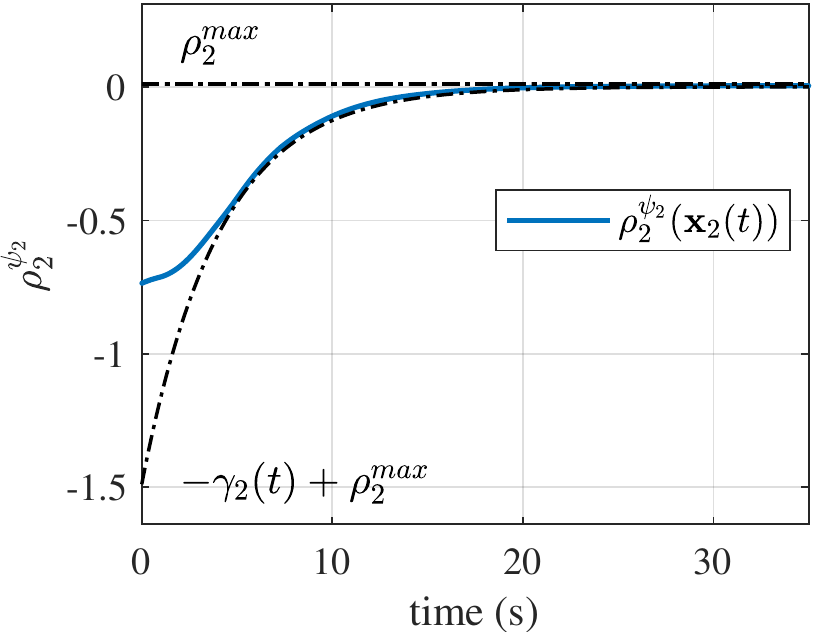}}
	\subcaptionbox{Funnel for $\Sigma_4$ with task $\psi_4$ \label{fig:robot_rho4}}
	{\includegraphics[width=.238\textwidth]{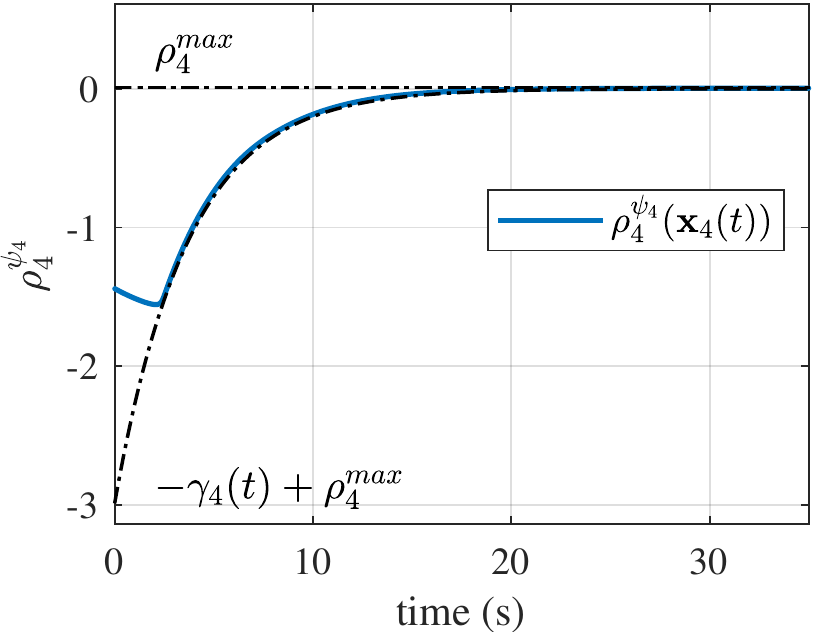}}
	
	
	\caption{Funnels for the local STL tasks. Performance bounds are indicated by dashed lines. Evolution of $\rho_i^{\psi_{i}}(\boldsymbol{x}_i(t))$ are depicted using solid lines.}
	\label{fig:robot}
	\vspace{-0.7cm}
\end{figure}

Let $\boldsymbol{p}_{i}:=[x_{i, 1}; x_{i, 2}]$ denote the position of of each robot $\Sigma_i$. 
The robots are subject to the following STL tasks:
 $\phi_1$: $F_{[0,35]}G_{[30,35]}((\Vert \boldsymbol{p}_{1}-[0.7;0.6]\Vert \leq 0.05) \wedge(\vert \deg(x_{1,3}) - 0 \vert \leq 7.5)$, 
$\phi_2$: $F_{[0,35]}G_{[30,35]}((\Vert \boldsymbol{p}_{2}-[0.725; 0.45]\Vert \leq 0.5) \wedge(\vert \deg(x_{2,3}) + 90 \vert \leq 7.5)$, $\phi_3$: $F_{[0,35]}G_{[30,35]}((\Vert \boldsymbol{p}_{3}-[0.5;0.425]\Vert \leq 0.5) \wedge(\vert \deg(x_{1,3}) + 180 \vert \leq 7.5)$, 
$\phi_4$: $F_{[0,35]}G_{[30,35]}((\Vert \boldsymbol{p}_{4}-[0.475; 0.55]\Vert \leq 0.5) \wedge(\vert \deg(x_{2,3}) - 145\vert \leq 7.5)$, $\phi_5$: $F_{[0,35]}G_{[30,35]}((\Vert \boldsymbol{p}_{5}-[0.575; 0.65]\Vert \leq 0.5) \wedge(\vert \deg(x_{2,3}) - 45\vert \leq 7.5)$, 
where $\deg(\cdot)$ converts angle units from radians to degrees.
Intuitively, each robot is assigned to move to its predefined goal point and stay there within the desired time interval, in the meanwhile satisfying the additional requirements on the robots' orientation. 
Next, we apply the proposed funnel-based feedback controllers as in \eqref{controller} to enforce the STL tasks on the 5-robot network. 
Numerical implementations were performed using MATLAB on a computer with a processor Intel Core i7 3.6 GHz CPU. Note that the computation of local controllers took on average 0.01 ms, which is negligible since $\boldsymbol{u}_{i}$ is given by a closed-form expression. Simulation results are shown in Figs.~\ref{fig:robot_position} and \ref{fig:robot}. The state trajectories of each robot are depicted as in Fig.~\ref{fig:robot_position} on the position plane. The triangles are used to indicate the dynamical evolution of the orientation of each robot.  
In Fig.~\ref{fig:robot}, we present the temporal behaviors of $\rho_i^{\psi_{i}}(\mathbf{x}_i(t))$ for two robots $\Sigma_2$ and $\Sigma_4$.
It can be readily seen that the prescribed performances of   $\rho_i^{\psi_{i}}(\mathbf{x}_i(t))$ are satisfied with respect to the error funnels, which shows that the time bounds are also respected.
Note that temporal behaviors of $\rho_i^{\psi_{i}}(\mathbf{x}_i(t))$, $i \in \{1,3,5\}$, also satisfy their prescribed performance bounds, although the figures are omitted here due to lack of space. 
Remark that the design parameters of the funnels are chosen according to the instructions listed in \eqref{funnelpara1}-\eqref{funnelpara2}, which guarantees the satisfaction of temporal formulae $ \rho_i^{\phi_i}(\mathbf{x}_i,0)$ by prescribing temporal behaviors of atomic formulae $\rho_i^{\psi_{i}}(\mathbf{x}_i(t))$, as shown in Fig.~\ref{fig:robot}. 
We can conclude that all STL tasks are satisfied within the desired time interval.

\vspace{-0.15cm}

\section{Conclusions}
We proposed a compositional approach for the synthesis of a fragment of STL tasks for continuous-time interconnected systems using assume-guarantee contracts.
A new concept of contract satisfaction, i.e., uniform strong satisfaction, was introduced to establish our contract-based compositionality result. A continuous-time feedback controller was designed to enforce the uniform strong satisfaction of local contracts by all  subsystems, while guaranteeing the satisfaction of global STL for the interconnected system based on the proposed compositionality result.

	
	

\vspace{-0.15cm}
\bibliographystyle{IEEEtran}
\bibliography{biblio}

\end{document}